\documentclass[a4paper,UKenglish,cleveref, autoref, thm-restate]{lipics-v2021}
\usepackage{xspace}
\newcommand{\ecalculus}{$\forall$\textsf{Exp+Res}\xspace}
\newcommand{\irc}{\textsf{IR-calc}\xspace}
\newcommand{\irmc}{\textsf{IRM-calc}\xspace}

\newcommand{\qrc}{\textsf{Q-Res}\xspace}

\newcommand{\qurc}{\textsf{QU-Res}\xspace}

\newcommand{\qrat}{\textsf{QRAT}\xspace}

\newcommand{\lqrc}{\textsf{LD-Q-Res}\xspace}
\newcommand{\mergeres}{\textsf{MRes}\xspace}
\newcommand{\qcp}{\textsf{CP+}$\forall$\textsf{red}\xspace}
\newcommand{\qef}{\textsf{eFrege+}$\forall$\textsf{red}\xspace}

\usepackage{tikz}
\usetikzlibrary{shapes,snakes,automata, matrix, arrows.meta, positioning}

\usetikzlibrary{fadings}
\tikzstyle{uedge}=[draw=blue!50!red]
\tikzstyle{fedge}=[draw=blue]
\tikzstyle{iedge}=[draw=red]
\tikzstyle{redge}=[draw=green!50!black]
\tikzstyle{rnode}=[draw,inner sep=2pt,color=black]
\tikzstyle{tnode}=[circle,minimum width=3pt,fill,inner sep=0pt]
\tikzstyle{dotnode}=[circle,minimum width=2pt,fill,inner sep=0pt]
\tikzstyle{labn}=[font=\sffamily,circle,fill=white,inner sep=1pt,draw=black]
\tikzstyle{legn}=[font=\scriptsize]
\tikzstyle{reln}=[circle,fill=white,inner sep=.4pt,draw=black]
\tikzstyle{oreln}=[circle,fill=white,inner sep=.4pt,draw=black!50,solid]
\tikzstyle{oree}=[thick,draw=black!50,densely dashed]
\tikzstyle{ree}=[thick,draw=black]
\tikzstyle{calcn}=[rectangle%
                   ,rounded corners=5mm%
                   ,thick %
                   ,draw=gray%
                   ,%top color=white,bottom color=black!20,%
                   minimum height=1cm,minimum width=1cm,inner sep=.2cm,%
                   text centered,font=\bfseries,%
                 font={\fontsize{12}{0}\selectfont\sffamily}
                 ]
\tikzstyle{expcalcn}=[rectangle%
                   ,rounded corners=1mm%
                   ,thick%
                   ,draw=gray, densely dashed%
                   ,%top color=white,bottom color=green!20,%
                   minimum height=1cm,minimum width=1cm,inner sep=.2cm,%
                   text centered,font=\bfseries,%
                 font={\fontsize{12}{0}\selectfont\ttfamily}
                 ]
\tikzstyle{expcalcn2}=[rectangle%
                   ,rounded corners=3mm%
                   , thick%
                   ,draw=gray,  densely dotted%
                   ,%top color=white,bottom color=green!20,%
                   minimum height=1cm,minimum width=1cm,inner sep=.2cm,%
                   text centered,font=\bfseries,%
                 font={\fontsize{12}{0}\selectfont\ttfamily}
                 ]
\tikzstyle{Style2}=[rectangle%
                   ,rounded corners=1mm%
                   ,ultra thick%
                   ,draw=gray%
                   ,%top color=white,bottom color=green!20,%
                   minimum height=1cm,minimum width=1cm,inner sep=.2cm,%
    align=center,font=\bfseries,
                   line width=2,%
                 font={\fontsize{12}{12}\selectfont\ttfamily}
                 ]
\tikzset{Style1/.style={black, draw=gray, minimum size=1cm}}
\usepackage{caption}
\usepackage{subcaption}

\title{QRAT Polynomially Simulates Merge Resolution.} %TODO Please add

\titlerunning{QRAT Polynomially Simulates Merge Resolution.} %TODO optional, please use if title is longer than one line

\author{Sravanthi Chede}{Department of Computer Science and Engineering, IIT Ropar, India \\sravanthi.20csz0001@iitrpr.ac.in }{}{https://orcid.org/0000-0001-7170-6156}{}%TODO mandatory, please use full name; only 1 author per \author macro; first two parameters are mandatory, other parameters can be empty. Please provide at least the name of the affiliation and the country. The full address is optional

\author{Anil Shukla}{Department of Computer Science and Engineering, IIT Ropar, India\\ anilshukla@iitrpr.ac.in}{}{}{}

\authorrunning{S. Chede and A. Shukla} %TODO mandatory. First: Use abbreviated first/middle names. Second (only in severe cases): Use first author plus 'et al.'

\Copyright{Sravanthi Chede and Anil Shukla} %TODO mandatory, please use full first names. LIPIcs license is "CC-BY";  http://creativecommons.org/licenses/by/3.0/

 %TODO mandatory: Please choose ACM 2012 classifications from https://dl.acm.org/ccs/ccs_flat.cfm 

 %TODO mandatory; please add comma-separated list of keywords

\category{} %optional, e.g. invited paper

\relatedversion{} %optional, e.g. full version hosted on arXiv, HAL, or other respository/website
\nolinenumbers %uncomment to disable line numbering

\hideLIPIcs  %uncomment to remove references to LIPIcs series (logo, DOI, ...), e.g. when preparing a pre-final version to be uploaded to arXiv or another public repository

\begin{document}

\maketitle

\begin{abstract}
Merge Resolution (\mergeres~\cite{mres_paper}) is a refutational proof system for quantified Boolean formulas (QBF). Each line of \mergeres consists of clauses with only existential literals, together with information of countermodels stored as merge maps. As a result, \mergeres has strategy extraction by design. The \qrat~\cite{qrat_paper} proof system was designed to capture QBF preprocessing. QRAT can simulate both the expansion-based proof system \ecalculus and CDCL-based QBF proof system \lqrc. 

A family of false QBFs called SquaredEquality formulas were introduced in~\cite{mres_paper} and shown to be easy for \mergeres but need exponential size proofs in \qrc, \qurc, \qcp, \ecalculus, \irc and reductionless \lqrc. As a result none of these systems can simulate \mergeres. In this paper, we show a short QRAT refutation of the SquaredEquality formulas. We further show that QRAT strictly p-simulates \mergeres. 
%This work not only shows the power of \qrat but also presents first simulation result of the \mergeres proof system. 
Besides highlighting the power of \qrat system, this work also presents the first simulation result for \mergeres. 
%In this paper we show that the family of equations which are hard for most proof systems except for M-Res \cite{hard_mres} i.e SquaredEquality Formulas are also easy for \qrat proof system. Furthermore, we show that \qrat polynomially simulates M-Res proof system.
\end{abstract}
\section{Introduction}
Quantified Boolean formulas (QBF) extend propositional logic with quantifications, there exists ($\exists$) and for all ($\forall$).
QBF proof complexity deals with understanding the limitations and strength of various QBF solving approaches. In the literature, there exists mainly two solving approaches i.e. Conflict-Driven-Clause-Learning (CDCL) and expansion-based solving. Several QBF proof systems have been developed to capture these solving approaches. Q-resolution (\qrc)~\cite{KBKF95} is the base of CDCL-based approach. It is further extended to QU-resolution~(\qurc)~\cite{qures_simul} and Long-Distance-resolution (\lqrc)~\cite{ldqres_paper}. On the other hand, proof system \ecalculus~\cite{JM15} is the base of expansion-based solving. It is further extended to powerful proof systems \irc~\cite{BCJ14} and \irmc~\cite{BCJ14}. The simulation orders of these proof systems are well studied in the literature~\cite[Figure 1]{BCJ15}.

Recently, a new proof system Merge resolution (\mergeres)~\cite{mres_paper} has been developed. It follows a different QBF-solving approach. In \mergeres, winning strategies for the universal player are explicitly represented within the proof in the form of deterministic branching programs, known as merge maps~\cite{mres_paper}. \mergeres builds partial strategies at each line of the proof such that the strategy at the last line (corresponding to the empty clause) forms the complete countermodel for the input QBF. As a result, \mergeres admits strategy extraction by design. While performing resolution steps, \mergeres merges the partial strategies of the two hypotheses carefully if their corresponding merge maps are isomorphic or consistent. Note that whether two merge maps are isomorphic or consistent can be checked efficiently. This allows those resolution steps to be performed in \mergeres which would have been blocked in \lqrc. 

To be precise, in \lqrc universal variables $u$ could appear in both polarities in the hypotheses and get merged in the resolvent provided $u$ appears in the right of the pivot variable in the quantifier prefix. \mergeres relaxed this restriction by allowing resolution steps even if $u$ is on the left of the pivot variable provided the merge maps of $u$ in both the hypotheses are isomorphic. This makes \mergeres powerful as compared to reductionless \lqrc~\cite{BjornerJK15,PeitlSS19a}. In fact there exists a family of false QBFs SquaredEquality formulas (Definition~\ref{def:squared}) with short refutations in \mergeres~\cite{mres_paper} but require exponential size refutations in \qrc, \qurc, \qcp~\cite{qcp18}, \ecalculus, \irc~\cite{BeyersdorffB20,BeyersdorffBH19} and reductionless \lqrc~\cite{mres_paper}. Therefore, none of these proof systems can simulate \mergeres. 

Quantified Resolution Asymmetric Tautologies (\qrat) proof system is introduced in~\cite{qrat_paper} to capture the preprocessing steps performed by several QBF-solvers. It has been shown in~\cite{qrat_paper} that \qrat can efficiently simulate all the existing preprocessing steps used by present-day QBF solvers. Recently, it has been shown that \qrat can simulate both the expansion-based proof system \ecalculus~\cite{ecalculus_simul} and CDCL-based proof system \lqrc~\cite{ldqres_simul}. Since \qrat allows resolution steps with universal variables as pivot, it simulates \qurc as well~\cite{ldqres_simul}. It is also known that \qrat is strictly stronger than \ecalculus, \lqrc and \qurc~\cite[Figure 2]{ldqres_simul}.

In this short paper, we extend the importance of \qrat among QBF proof systems by showing that \qrat even polynomially simulates \mergeres. We also show that refuting the SquaredEquality formulas in \qrat is easy. Thus the semantic structure of these formulas which makes it harder to refute in all other proof systems is not a restriction for \qrat. We explain these contributions in the following subsection.

%The \qrat system simulates everyone of them\cite{qures_simul}\cite{qres_simul}\cite{ldqres_simul}.

%On the other hand, \ecalculus is the base for expansion based solving and it has been proved to be simulated by \qrat~\cite{ecalculus_simul}. \ecalculus can be extended to \irc and it is still open if \qrat simulates \irc?~\cite{our_paper}. \mergeres being a new approach in QBF solving (Strategy Extraction) and \qrat being a powerful proof system which simulates most proof systems, it is interesting to check if \qrat can also simulate

\subsection{Our contributions}
\begin{alphaenumerate}
    \item {\bf Short QRAT refutation of SquaredEquality formulas:} SquaredEquality formulas, a variant of equality formulas~\cite{BeyersdorffBH19}, have been defined in~\cite{mres_paper} to show that \mergeres is strictly stronger than reductionless \lqrc~\cite{BjornerJK15,PeitlSS19a}. The original equality formulas which are hard for \qrc but easy for \lqrc have been extended in a way that prohibits the resolution step in reductionless \lqrc but not in \mergeres. 
    \vspace{0.1cm}\\
    In this paper, we show that the SquaredEquality formulas have a short refutation in \qrat (Theorem~\ref{short_proof}). Also, since the original equality formulas are easy for \lqrc and \qrat can simulate \lqrc, the formulas are easy for \qrat as well. 
    \vspace{0.1cm}\\
    Note that, all other known families of false QBFs used to establish the incomparability results among QBF proof systems are easy for \qrat: KBKF~\cite[Theorem 3.2]{KBKF95} formulas are easy for \qurc~\cite[Example 5.5]{qures_simul} hence they are easy for \qrat. Similarly, QPARITY~\cite{BCJ15} formulas are easy for \ecalculus~\cite[Lemma 15]{BCJ15} and hence easy for \qrat. Variants of these formulas were used to show the incomparability results among proof systems known to be simulated by \qrat. Hence these formulas are also easy for \qrat.\
    %Also, all variants of these formulas are easy for \qrat. 
    Thus the presented short \qrat refutation of the SquaredEquality formulas makes this formulas also easy for \qrat.
    %is a step ahead in \qrat being the all powerful proof system. %that it enacts to be.%this direction. 
    
    \vspace{0.1cm}
    \item {\bf \qrat polynomially simulates \mergeres:} It has been shown that \mergeres can simulate reductionless \lqrc~\cite{mres_paper}. However, none of the proof systems \qrc, \qurc, \qcp, \ecalculus, \irc and reductionless \lqrc are capable of simulating \mergeres. The difficulty for these proof systems lies in simulating the axiom steps of \mergeres. To be precise, \mergeres gets rid of all the universal variables from the input clauses just by maintaining the partial strategies for them. On the other hand, the above mentioned proof systems have different and restricted rules for handling the universal variables. For example, \qrc and \qurc use universal reduction(UR) rule, which allows dropping a universal variable only if it is not blocked. Similarly, expansion-based proof systems like \ecalculus and \irc handle the universal variables by introducing the annotated existential variables only. 
    \vspace{0.1cm}\\
    In this paper, we show how \qrat handles this hurdle and polynomially simulates \mergeres. We show this by proving that the downloaded clauses in \mergeres proofs are all Asymmetric Tautology (AT) (Definition \ref{asymmetric-tautology}) with respect to the input QBF (Lemma \ref{existential_clause}). Therefore, they can easily be added in \qrat. Since the resolution step can be easily simulated by \qrat (Observation~\ref{obs:qrat-simulate-res}), the remaining resolution steps in \mergeres refutation can also be simulated (Theorem \ref{simul}).
    
    %Consider a simple example. Let $\mathcal{F} = \exists x \forall y \exists z. (x \vee y \vee z) \wedge \dots$ be a QBF. \mergeres may download the following clause $(x \vee z)$ and maintains a partial strategy $(z \rightarrow 0)$. 
    \vspace{0.1cm}
    \item {\bf Emphasizing the importance of \qrat among QBF proof systems:} \qrat has been shown to simulate varieties of QBF solving approaches. That is on one hand, \qrat can simulate the expansion-based system \ecalculus and on the other, it can simulate the powerful CDCL-based system \lqrc. Since \mergeres is based on an entirely different QBF-solving approach; by showing that \qrat can polynomially simulate \mergeres, the paper extends the importance of \qrat system.  
    
    \qrat is a possible candidate for the universal checking format which can verify all existing QBF-solving techniques~\cite{chew21}. Our simulation result is a small step in this direction. (The other possible candidate is the extended Frege for QBFs, denoted as, \qef~\cite[Conjecture 1]{chew21}).
    For the simulation order and incomparabilities involving \qrat and several QBF proof systems, refer Figure~\ref{fig1}.
    \end{alphaenumerate}
    \vspace{0.25cm}
    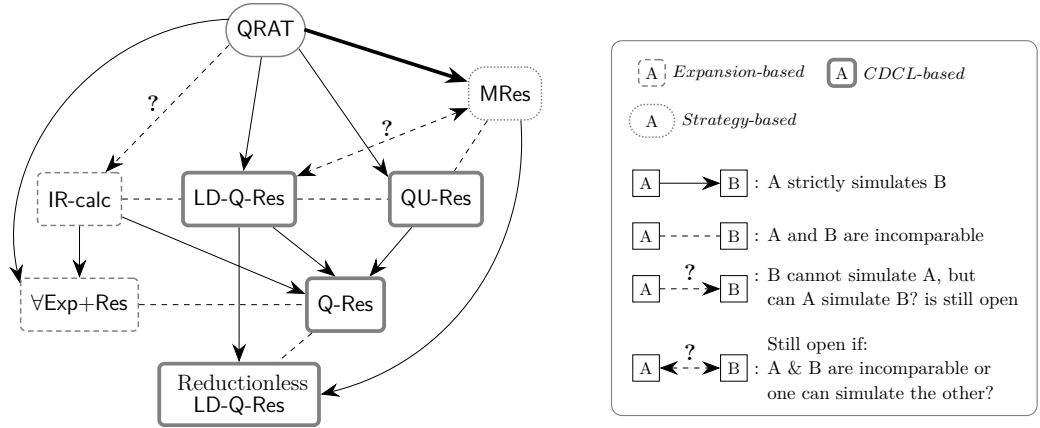
\begin{figure}[!h] %[15]{r}{0.53\textwidth}
    \centering
  \begin{tikzpicture}[scale=0.7]
  \draw[white, ultra thin,rounded corners=2mm] (-6.6,- 2.6) rectangle (13,6);
  \pgflowlevelsynccm
  %\node[Style2] at (0,0) {Exp-based};
  \node[expcalcn] (eres) at (-5,0) {\ecalculus};
  \node[expcalcn] (irc) at (-5,2) {\irc};
  %\node[expcalcn] at (-3,-3) {CDCL-based};
  \node[Style2] (qurc) at (1.7,2) {\qurc};
  \node[Style2] (lq) at (-2,2) {\lqrc};
  \node[Style2] (red-lq) at (-2,- 1.7) {\textnormal{ Reductionless} \\ \lqrc};
  \node[Style2] (qrc) at (0,0) {\qrc};
  \node[calcn] (qrat) at (- 1.5,5.2) {\qrat};
  \node[expcalcn2] (mres) at (3,4) {\mergeres};
  %\draw[black,dashed, ->] (g.south) edge[bend right]  (b.west);
  \draw[black, -{Stealth[scale=2]},shorten <= -2pt] (qrat.165) to [bend right=50] (eres.160);
  \draw[black, -{Stealth[scale=2]}] (irc) to  (eres);
  \draw[black, -{Stealth[scale=2]}] (qrat) to  (lq);
  \draw[black, -{Stealth[scale=2]}, shorten <= -5.5pt] (qrat.325) to  (qurc.155);
  \draw[black, -{Stealth[scale=2]}] (lq.320) to  (qrc.110);
  \draw[black, -{Stealth[scale=2]}] (qurc) to  (qrc);
  \draw[black, -{Stealth[scale=2]}] (irc) to  (qrc.165);
  \draw[black, -{Stealth[scale=2]}] (lq) to  (red-lq);
  \draw[black, -{Stealth[scale=2]}] (mres.300) to [bend left=40] (red-lq.east);
  \draw[black, -{Stealth},line width=0.7mm,shorten <= -1pt] (qrat.east) to  (mres);
  \draw[black, dashed , -] (eres) to  (qrc);
  \draw[black, dashed , -] (irc) to  (lq);
  \draw[black, dashed , -] (lq) to  (qurc);
  \draw[black, dashed , -] (mres) to  (qurc);
  \draw[black, dashed , -] (qrc) to  (red-lq);
  \draw[black,dashed, -{Stealth[scale=2]},  shorten <= -3pt] (qrat.205) -- (irc)node[pos=.5,above=4pt, left] {\Large\textbf{?}};
  \draw[black,dashed, {Stealth[scale=2]}-{Stealth[scale=2]}] (mres) -- (lq)node[pos=.4,above=4pt, left] {\Large\textbf{?}};
  
  \draw[gray,semithick,rounded corners=2mm] (5, 5) rectangle (13,- 2.1);
  \matrix [matrix of math nodes,row sep=3mm,column sep=3mm, nodes={font=\small},below right] at (5.2,4.8) %(current bounding box.north east)
{
\node [expcalcn,minimum height=.5cm, minimum width=.5cm, inner sep=.1cm, font=\small, label=right:Expansion {\text - } based] {\textnormal{A}}; & \node [Style2,minimum height=.5cm, minimum width=.5cm, inner sep=.1cm, font=\small,label=right:CDCL{\text - } based] {\textnormal{A}}; \\
\node [expcalcn2,minimum height=.5cm, minimum width=.5cm, inner sep=.2cm, font=\small, label=right:Strategy {\text - } based] {\hspace{0.1cm}\textnormal{A}\hspace{0.1cm}};\\
};
\node[rectangle,draw=black,minimum height=.5cm, minimum width=.5cm, inner sep=.1cm, font=\small] (a) at (5.65,2.3) {A};
\node[rectangle,draw=black,minimum height=.5cm, minimum width=.5cm, inner sep=.1cm, font=\small, label=right:{: A strictly simulates B}] (b) at (7.3,2.3) {B};
\draw[black, -{Stealth[scale=2]}] (a) to (b);
\node[rectangle,draw=black,minimum height=.5cm, minimum width=.5cm, inner sep=.1cm, font=\small] (c) at (5.65,1.3) {A};
\node[rectangle,draw=black,minimum height=.5cm, minimum width=.5cm, inner sep=.1cm, font=\small, label=right:{: A and B are incomparable}] (d) at (7.3,1.3) {B};
\draw[black, dashed , -] (c) to (d);
\node[rectangle,draw=black,minimum height=.5cm, minimum width=.5cm, inner sep=.1cm, font=\small] (e) at (5.65,0.3) {A};
\node[draw=black, rectangle, minimum height=.5cm, minimum width=.5cm, inner sep=.1cm, font=\small, label={[align=left] right:: B cannot simulate A, but \\\hspace{0.25cm}can A simulate B? is still open}] (f) at (7.3,0.3) {B};
\draw[black,black,dashed, -{Stealth[scale=2]}] (e) -- (f)node[pos=.5,above=1pt] {\Large\textbf{?}};
\node[rectangle,draw=black,minimum height=.5cm, minimum width=.5cm, inner sep=.1cm, font=\small] (g) at (5.65,-1.2) {A};
\node[rectangle,draw=black, minimum height=.5cm, minimum width=.5cm, inner sep=.1cm, font=\small, label={[align=left]0: \hspace{0.25cm}Still open if:\\: A \& B are incomparable or\\ \hspace{0.25cm}one can simulate the other?}] (h) at (7.3,-1.2) {B};
\draw[black,dashed, {Stealth[scale=2]}-{Stealth[scale=2]}] (g) -- (h)node[pos=.5,above=1pt] {\Large\textbf{?}};
\end{tikzpicture}
\vspace{0.2cm}
  \caption{Simulation order of QBF proof systems, with our new result shown in bold. \mergeres simulation of reductionless \lqrc is shown in \cite{mres_paper}. \qrat simulation of \ecalculus, \lqrc and \qurc are shown in \cite{ecalculus_simul,ldqres_simul,ldqres_simul} respectively. The incomparability result of reductionless \lqrc and \qrc was shown in \cite{PeitlSS19a}. For other known relations refer~\cite[Figure 1]{BCJ15}}\label{fig1} 
\end{figure}
\vspace{-0.7cm}
\subsection{Organisation of the paper}
\vspace{-0.1cm}
In Section~\ref{sec:Prerequisites}, we denote all important notations and preliminaries used in the paper. We define \mergeres in Section~\ref{2.1} and \qrat in Section~\ref{2.2}.
In Section~\ref{3}, we define the SquaredEquality formulas and give a short \qrat proof for the same in (Theorem~\ref{short_proof}).
We prove the \qrat simulation of \mergeres in Section~\ref{4}. Finally, we conclude and discuss further open problems in Section~\ref{5}. % and mention further open problems related to this work.
\vspace{0.5cm}
\section{Notations and Prerequisites}\label{sec:Prerequisites}
A clause $C$ is a disjunction of literals and a conjunctive normal form (CNF) $F$ is a conjunction of clauses. A clause $D$ is a sub-clause of $C$ if every literal of $D$ are also a literal of $C$.
A unit clause is a clause with only one literal in it. We denote the empty clause by $\bot$. vars($C$) is a set of all variables in $C$ and $var(F) = \cup_{C\in F}~vars(C)$. $\overline{C}$ for a clause $C$, is a conjunction of negation of literals in $C$.
\vspace{0.2cm}\\
%In this paper, we consider proofs as a sequence of clauses and a refutation is a proof deriving $\bot$ in the end.\\
A \textbf{proof system}~\cite{CookR79} for a non-empty language $L \subseteq \{ 0, 1\}^*$ is a
polynomial time computable function $f : \{ 0, 1\}^* \rightarrow \{ 0, 1\}^*$ such that Range($f$) = $L$. For string $x \in L$, we say a string $w \in \{ 0, 1\}^*$ is an $f$-proof of $x$ if $f(w)$ = $x$. 
A proof system $f$ for $L$ is complete iff for every $x \in L$ we have a corresponding $f$-proof for $x$. A proof system $f$ for $L$ is sound iff the existence of an $f$-proof for $x$ implies that $x \in L$. 
%We say a proof system for $L$ is polynomially bounded if there exists a polynomial $p(x) \in \mathbb{N}[x]$ such that each $x \in L$ has an $f$-proof `$w$' of size $|w| \leq p(|x|)$.

%Informally, a proof system is a function $f$ which maps proofs to theorems (or contradictions). 
A proof system $f$ p-simulates (polynomially simulates) another proof system $g$ (i.e., $f \leq_p g$) if both prove the same language $L$ and every $g$-proof of input $x \in L$ can be translated into an $f$-proof for the same input in time that is polynomial w.r.t size of the $g$-proof. Otherwise, we say that $f$ do not simulate $g$ ($f \not \leq_p g$). We say that a proof system $f$ strictly simulates a proof system $g$ if $f \leq_p g$ but $g \not \leq_p f$. 
Proof systems $f$ and $g$ are said to be incomparable, if none of them can simulate the other. $f$ and $g$ proof systems are said to be p-equivalent if both $f \leq_p g$ and $g \leq_p f$ hold.\vspace{0.2cm}\\
Proof systems for $L = UNSAT/TAUT$ are called propositional proof systems. For example, the resolution proof system is defined as follows:
\vspace{-0.2cm}
\begin{definition}\textbf{Resolution proof system:} \label{resolution} Resolution proof system~\cite{Bla37,Rob65} is the most studied propositional proof system. The lines in this proof system are clauses. Given a CNF formula $F$, it can derive new clauses using the following inference rule, also known as the resolution rule: $\frac{C \vee x \hspace{5mm} D \vee \neg x}{C \vee D},$ where $C$ and $D$ are clauses and $x$ is the pivot variable being resolved. The clause $C \vee D$ is called the resolvent. For the rest of this paper we denote this step as Res($(C \lor x), (D \lor \overline{x}), x$).
\end{definition} 
\vspace{-0.2cm}
%In this paper, we consider proofs as a sequence of clauses/lines and a refutation is a proof deriving $\bot$ in the end.\\
Proof systems for $L = \text{FQBFs/ TQBFs}$ are said to be QBF proof systems where, FQBFs (TQBFs) denote the set of all false (true) QBFs.
For example, \qrc, \qurc, etc. Before defining them we first define the QBFs.
\vspace{0.2cm}\\
{\bf Quantified Boolean formulas:} QBFs are an extension of the propositional Boolean formulas where each variable is quantified with one of $\{ \exists,\forall \}$, the symbols having their general semantic definition of existential and universal quantifier respectively.

%Generally represented by $Q.\phi$, the propositional part $\phi$ of a QBF is called the matrix which should be in CNF (Conjunctive Normal Form) and the rest is the quantification sequence of variables called the prefix $Q$. 
In this paper, we assume that QBFs are in closed prenex form i.e., we consider the form $Q_1 X_1 . . .Q_k X_k$. $\phi(X)$ , where $X_i$ are pairwise disjoint sets of variables; $Q_i$ $\in$ \{$ \exists$, $ \forall$\} and $Q_i \neq Q_{i+1}$, and $\phi(X)$ is in CNF form over $X = X_1 \cup \dots \cup X_k$, called the matrix of the QBF. We denote QBFs as $Q.\phi$ in this paper, where $Q$ is the quantifier prefix.   

If a variable $x$ is in the set $X_i$, we say that $x$ is at level $i$ and write $lv(x)=i$. Note that the quantifier of a literal is the quantifier of the corresponding variable.
%Note that only variables are in prefix, so given a literal $l$ if $x$=var($l$),then quantifier of $l$ is same as that of $x$.
Given a literal $\ell$ with quantifier $Q_i$ and a literal $k$ with quantifier $Q_j$ , we say that $\ell$ occurs to left of $k$ and write $\ell  {\leq}_Q k$ if $i \leq j$. Likewise, we say that $\ell$ occurs to right of $k$ and write $\ell  {>}_Q k$ if $i > j$. %Within the set $X_i$ we extend the ordering as the ordering of variables is not

%QBFs are often seen as a game between the universal and the existential player i.e in the $i^{th}$ step the player $Q_i$ assigns values to the variables $X_i$. The existential (resp universal) player wins if substituting this total assignment of variables in $\phi$ evaluates to $1$ (resp 0).

An assignment tree of a QBF $Q.\phi$ is a complete binary tree of depth $|vars(\phi)|$ where each level is dedicated to a variable in the order of the prefix. If the variable is existential, then the level is said to be an existential level and each node of this level is said to be an existential node. Similarly we have universal levels and universal nodes.
If a level is dedicated for a variable (say $x$), every node in that level will have 2 outgoing edges for the level below marked with $x$ and $\overline{x}$ respectively except for the level corresponding to leaf nodes. 
Each node is labelled with either $1$ or $0$ eventually as follows:  
the path from root to leaf is a total assignment to all variables which if evaluate $\phi$ to $1$ (resp.~$0$), the corresponding leaf node is labelled with $1$ (resp.~$0$). Rest of the nodes are labelled from bottom up, where the existential nodes act as $OR$ gates and the universal nodes act as $AND$ gates. 
A \textbf{model} (resp.~\textbf{countermodel}) is a sub-tree of the assignment tree where every existential node has exactly one child (resp.~both children) and universal node has both the children (resp.~exactly one child), furthermore every node in this tree is marked with $1$ (resp.~$0$). True-QBFs have at least one model and false-QBFs have at least one countermodel.
%A QBF is called false if it has a countermodel, and true if it does not.
%The term 'winning strategy' is often used instead of model(in case of existential player) and countermodel(in case of universal player). The term 'strategy extraction' means extracting the assignments to universal(resp existential) variables based on existential(resp universal) variables from the refutation(resp proof) or countermodel(resp model) of a false-QBF(resp True-QBF).  A QBF is false if there is always a \textit{winning strategy} for the universal player for any assignments given to existential variables.
%\noindent
%{\bf Proof system}~ for a language $L$ is a polynomial time verifier algorithm $V$ such that for all inputs $x$, $x \in L$ iff there exists a proof $P$ such that $V$ accepts input ($x,P$). The complexity of a proof system is a measure of how large $|P|$ has to be as a function of $x$.
\vspace{0.2cm}\\
%{\bf Winning strategy and Strategy extraction:}
{\bf QBFs as a game:} 
QBFs are often seen as a game between the universal and the existential player i.e. in the $i^{th}$ step the player $Q_i$ assigns values to the variables $X_i$. At the end, the existential (resp.~universal) player wins if substituting this total assignment of variables in $\phi$ evaluates to $1$ (resp.~0).

For a QBF $Q.\phi$, a \textbf{strategy} of universal (resp.~existential) player is a decision function that returns the assignment to all universal (resp.~existential) variables of $Q$, where the decision for each $u$ depends only on the variables to the left of it in the quantifier prefix $Q$. 

\textbf{Winning strategy} for a player is a strategy which makes this player win against every assignment of the other player.% variables.
The term `winning  strategy' is often used instead of model (in case of the existential player) and countermodel (in case of the universal player).
A QBF is false (true) iff there exists a winning strategy for the universal (existential) player~\cite{AroraBarak09}.

%\textbf{Strategy extraction} is the property of any proof system which efficiently computes the winning strategy of the universal player w.r.t input QBF, from the proof in that system of the same input QBF.

We say that a QBF proof system $f$ admits \textbf{strategy extraction} if from every $f$-proof ($f$-refutation) of a true (false) QBF $Q.\phi$ one can extract a winning strategy for the existential (universal) player efficiently w.r.t. the size of the $f$-proof ($f$-refutation). 

%Strategy extraction is one of the distinctive features of QBF, with uses in 
%both QBF-solving and proof complexity. For QBF-solving, it helps to easily check the true/false output given by the solver. For
%proof complexity. It is used as a lower-bound technique exploiting the fact that formulas requiring hard strategies cannot have short proofs when efficient strategy extraction has been implemented for a proof system.\\
%In QBF proof complexity, strategy extraction technique is used extensively 
Strategy extraction is one of the most important lower bound techniques in QBF proof complexity. 
If a proof system $f$ admits strategy extraction, then every QBF with hard strategies must have long $f$-proofs.\vspace{0.2cm}\\
Now, let us define few important QBF-proof systems:\\
\textbf{\qrc:}~\cite{KBKF95} \qrc is the extension of the resolution proof system for QBFs. It has two rules namely resolution and universal reduction. The resolution rule is the same as defined before in (Definition~\ref{resolution}) i.e. Res($C_a,C_b,x$); the only restrictions being that the pivot variable $x$ should be an existential variable and that the resolvent clause should not be a tautology. 

The \textbf{Universal Reduction} (UR) of \qrc is the rule that allows dropping of universal literal $u$ from a clause $C$ in the QBF provided no existential literal $\ell \in C$ appears to the right of $u$ in the quantifier prefix. In this case, we say $u$ is not blocked in $C$, otherwise it is blocked. \\
\textbf{\qurc:}~\cite{qures_simul} \qurc is an extension of \qrc which allows resolution on universal variables as well.

\subsection{\mergeres proof system~\cite{mres_paper}}\label{2.1}
%The \mergeres proof system was introduced to provide linear strategy extraction in QBF-solving~\cite{mres_paper}. 
\mergeres is a proof system for false QBFs introduced in~\cite{mres_paper}.
We describe \mergeres briefly in this section, please refer to \cite{mres_paper} for its formal definition.

A \mergeres refutation of a QBF $Q.\phi$ consists of a sequence of lines where each line $L_i$ consists of a clause $C_i$ and a map (${M_i}^u$) for each universal variable $u \in Q$. At any given point in the proof, $C_i$ consists of only existential literals and ${M_i}^u$ gives the partial strategy of universal variable $u$ based on the existential variables which lie to the left of $u$ in the quantifier prefix $Q$.

The \textbf{Merge-maps} ${M_i}^u$ can be either directly $i \mapsto \{ u / \overline{u} / * \} $ or it can be of the form $i \mapsto (x,a,b)$ (read as ‘if $x = 0$ then goto $a$ else goto $b$’) where $x$ is an existential variable and $a,b < i$ are line indices from previous lines of the \mergeres proof. 

Merge-maps can be represented as graphs where the node labels are line indices and the edges are labelled by existential literals. For example, for the merge map rule $i \mapsto (x,a,b)$, we have an edge $i \rightarrow a$ in the graph with label $\overline{x}$ and an edge $i \rightarrow b$ with label $x$. 
%edge $i \rightarrow a$ of the graph is labelled with $\overline{x}$ corresponding to the rule $i \mapsto (x,a,b)$)
Let $i$ be a line index which is a leaf node in the graph, then 
%the leaf nodes of the graph 
it also has an additional label from $\{ u, \overline{u}, * \}$ corresponding to the rule $i \mapsto \{ u / \overline{u} / * \} $.\vspace{0.2cm}\\
The following two properties can be easily checked on merge-maps:

\textbf{Isomorphism:} Two merge maps ${M_a}^u$ and ${M_b}^u$ are isomorphic (written  ${M_a}^u \simeq {M_b}^u$) if and only if there exists a bijection mapping from the line numbers of one to those of another when represented as graphs. In other words, two isomorphic merge maps represent the same strategy.
%{\bf TODO: You need to define how merge maps can be represented as graphs since you are using the same}

\textbf{Consistency:} Two merge maps ${M_a}^u$ and ${M_b}^u$ are consistent (written ${M_a}^u \bowtie {M_b}^u$) if and only if for every common line index (say $i$) in both maps, it holds that  ${M_a}^u (i) = {M_b}^u (i)$.\vspace{0.2cm}\\
%\vspace{0.1cm}\\
The following two functions are defined on merge-maps:

\textbf{Select(${M_a}^u, {M_b}^u$):} It is defined only when ${M_a}^u \simeq {M_b}^u$ or when either of them is trivial (i.e. $a \mapsto *$ or/\& $b \mapsto *$ ). In such a case, it returns ${M_a}^u$ (if ${M_a}^u$ is not trivial), otherwise returns ${M_b}^u$.

%Let $M_1$ and $M_2$ be merge maps for which $M_1 \simeq M_2$ or one of $M_1, M_2$ is trivial. Then select($M_1, M_2$) := $M_2$ if $M_1$ is trivial, and select($M_1, M_2$) := $M_1$ otherwise.
%We call a merge map trivial iff it is isomorphic to $1 \rightarrow *$.
\textbf{Merge(${M_a}^u, {M_b}^u, n, x$):} It is defined only when ${M_a}^u \bowtie {M_b}^u$ and $x$ is an existential variable and $n$ is a new line index strictly greater than both $a$ and $b$. In such a case, it returns a new merge-map which merges the same indice nodes into one node and adds a new node with the rule $n \mapsto$ ($x,a,b$). Note that the indices present only in one of the input maps (i.e. not-common) are retained in the new merged map as they were.
\vspace{0.2cm}\\
%Let $M_1$ and $M_2$ be consistent merge maps for $u$ over $X$ with domains $N_1$ and $N_2$, let $n > max(N_1 \cup N_2)$ be a natural number, and let $x \in X$. Then merge($M_1, M_2, n, x$) is the function from $N_1 \cup N_2 \cup \{ n \}$ defined by\\
Now, we are ready to define the \mergeres proof system:
\vspace{-0.3cm}
\begin{definition}[\textbf{\mergeres proof system}~\cite{mres_paper}]\label{mres} %(\textbf{\mergeres})
Let $\Phi$ := $Q. \phi$ be a QBF with existential variables $X$ and universal variables $U$. \mergeres derivation of $\Phi$ is a sequence $\pi := L_1, . . . , L_k$ of lines $L_i := (C_i , \{ {M_i}^u : u \in U \} )$ derived by one of the following steps: 
\begin{alphaenumerate}

\item \textbf{Axiom}. There exists a clause in $C \in \phi$ such that $C_i$ is the existential sub-clause of $C$, and, for each $u \in U$, ${M_i}^u$ is
% the merge map for $u$ with 
the rule $i \mapsto$
%over $L_{Q(u)} with domain \{ i \}$ mapping $i$ to 
the falsifying $u$-literal for $C$, if $u \notin C$ add the trivial rule $i \mapsto *$; {\bf or},
\item \textbf{Resolution}. There exist integers $a, b < i$ and an existential pivot $x \in X$ such that $C_i$ = Res($C_a, C_b, x$), where one of the following must hold for every $u \in U$:
%and, for each $u \in U$, either
\begin{description}
\item[(i)] ${M_i}^u$ = select($ {M_a}^u , {M_b}^u)$ if defined; \textbf{or},
\item[(ii)] $x <_Q u$ and ${M_i}^u$ = merge($ {M_a}^u, {M_b}^u), i, x )$.\vspace{0.1cm}
\end{description}
\end{alphaenumerate}
The final line $L_k$ is the conclusion of $\pi$, and $\pi$ is a refutation of $\Phi$ iff $C_k = \bot$. In this case observe that $\{{M_k}^u : u \in U\}$ is a winning strategy for the universal player. %Size of $\pi$ = $k$.
\end{definition}
%As \mergeres has easy proof for SquaredEquality formulas which are hard for most proof systems; it is considered as a powerful proof system. In this paper we will show that \qrat is even more powerful that this system. 
It's known that \mergeres is sound and complete for false QBFs~\cite[Section 4.3]{mres_paper}.
We outline \qrat proof system in the next section.
\subsection{\qrat proof system~\cite{qrat_paper}}\label{2.2}
The \qrat proof system was introduced to capture the state-of-the-art techniques used in current day QBF-solvers~\cite{qrat_paper}. We give a brief summary of its rules. We need the following definitions:\vspace{-0.2cm}
%We need only the restricted version of \qrat which will be used in this paper (refer \cite{qrat_paper} for the total proof system)
\begin{definition}
For a CNF formula $F$, \textbf{unit propagation} (represented by $\vdash_1$ or unit-propagation$(F)$) simplifies $F$ on unit clauses; that is for every unit clause $(\ell) \in F$, it assigns $\ell$ to 1 in all clauses of $F$.  i.e. removes all clauses that contain the literal $\ell$ from the set $F$ and drops the literal $\overline{\ell}$ from all clauses in $F$. It keeps repeating this until no unit clause is left or an empty clause is derived.\vspace{-0.2cm}
\end{definition}
\begin{definition}[Outer resolvent~\cite{qrat_paper}]\label{outer-resolvent}
Given two clauses $(C \lor \ell), (D \lor \overline{\ell})$ of a QBF $Q.\phi$, the \textbf{Outer Resolvent} OR($Q$,$C$,$D$,$\ell$) is the clause consisting of all literals in $C$ together with those literals of $D$ that occur to the left of $\ell$, i.e. $C \cup \{ k~|~k \in D , k \leq_Q \ell \}$.\vspace{-0.2cm}
\end{definition}
\begin{definition}[Asymmetric Tautology (AT)]\label{asymmetric-tautology}
Clause $C$ is an AT w.r.t. to CNF $\phi$ iff $\phi$ $\vdash_1  C$. Alternatively, $C$ is an AT w.r.t. $\phi$ iff $\bot \in$ unit-propagation$(\phi \land \overline{C})$.
A clause $C$ is an AT w.r.t. a QBF $Q.\phi$ if it is an AT w.r.t.~$\phi$.
\end{definition}

%Unit propagation($\vdash_1$) simplies a CNF by repeating the following: If there is a unit clause ($\ell$) then remove all clauses that contain the literal $\ell$ and remove the literal $\overline{\ell}$ from all clauses.
%\vspace{0.1cm}
\noindent
\textbf{\qrat-clause \& \qrat-literal:} A clause $C\vee \ell$ is \qrat-clause w.r.t. a QBF $Q.\phi$ if for every $D\vee \overline{\ell} \in \phi$ the OR($Q$,$C$,$D$,$\ell$) is an AT w.r.t. $\phi$. We say that $\ell$ is the  \qrat-literal in $C$.
%implied by unit propagation. 

If a clause $C$ contains an existential \qrat-literal, it has been shown in \cite{qrat_paper} that $C$ can be removed (called QRATE rule) or added (called QRATA rule) without affecting the satisfiability of the QBF. Also, if a clause $C$ contains a universal \qrat-literal $\ell$, then dropping $\ell$ from $C$ (called QRATU rule) is also a satisfiability preserving step. Note that a clause $C$, which is an AT w.r.t.~a QBF $\Phi=Q.\phi$ is also a \qrat-clause on any literal belonging to $C$ w.r.t. $\Phi$~\cite{qrat_paper}. Additionally, \qrat allows elimination of any clause at any point in the proof~\cite{ecalculus_simul}.\vspace{0.2cm}\\
The remaining rule in \qrat is EUR; to define it we need the following:

\textbf{Extended inner clause (EIC):} For a QBF $Q.\phi$ where $C \in \phi$ and $\ell \in C$, EIC($Q, C, \ell$) is the final clause obtained when repeatedly performing the following: for every existential literal $k \in C$ which is to the right of $\ell$ in $Q$, extend $C$ by all the right literals of $\ell$ in the clauses $D \in \phi$ with $\overline{k} \in D$ (also include $\overline{\ell}$ in $C$ if $\overline{\ell}$ also $\in$ such $D$).

\textbf{Extended Universal Reduction (EUR):} Given a QBF $Q.\phi$, for a clause $C \in \phi$ with a universal literal $\ell \in C$ such that $\overline{\ell} \notin$ EIC($Q,C,\ell$), the literal $\ell$ can be dropped from $C$ under the Extended Universal Reduction (EUR) rule of \qrat. %For the details refer~\cite{qrat_paper}.
\vspace{0.2cm}\\
\noindent
Given a QBF $\Phi=Q.\phi$, a sequence of clauses is called a \qrat refutation of $\Phi$, if they are derived using the above mentioned rules and the last clause in the sequence is $\bot$.
\vspace{0.2cm}\\
\noindent
In this paper we show that even the full power of QRAT is not required to prove the SquaredEquality formulas or to simulate \mergeres. This restricted variant is referred to as \qrat(UR) in the literature~\cite{ChewC20} which allows all the \qrat rules but uses universal reduction (UR) instead of the powerful EUR rule. In fact, if we allow the definition that clauses with only universal literals are $\bot$; then \qrat simulation of \mergeres does not even require the {\bf UR} rule, only the QRATU rule is sufficient.

\vspace{0.2cm}
%We end this section, with the following observation which we need for the proof.
\noindent
Before moving on to the next section, we state the following observation which is useful for the upcoming proofs.
%Note that if this UR rule in \qrat(UR) is also not required for the simulation and  in this paper.
\vspace{-0.2cm}
\begin{observation}[\cite{qrat_paper}]\label{obs:qrat-simulate-res}
Consider a resolution step Res($(C\lor x),(D \lor \overline{x}),x$). \qrat can easily simulate the resolution steps by directly adding the resolvent clause to the QBF $Q.\phi$ as it is an AT w.r.t. $\phi$. This is true as unit propagation of the resolvent $(C \lor D)$ in $\phi$ derives $(x \wedge \overline{x}) = \bot$. 
Note that the above argument is valid for universal pivot variables as well. This implies, \qrat can simulate \qurc. 

%Note that the pivot variable $x$ is allowed to be universal variable, hence \qrat simulates \qurc as well.
\end{observation}
\section{SquaredEquality Formulas~\cite{mres_paper}}\label{3}
In \cite{hard_mres}, it was stated that, there exists a family of false QBFs, the SquaredEquality formulas, with short proofs in \mergeres but requiring exponential size in \qrc, \qurc, \qcp, \ecalculus, \irc and reductionless \lqrc. 
%Nothing has been said about the \qrat proof size for these formulas. 
In the next subsection, we show that these formulas are easy for \qrat proof system as well. We next present its definition.
\begin{definition}[SquaredEquality Formulas~\cite{mres_paper}]\label{def:squared} The squared equality family is the QBF family whose $n^{th}$ instance ${EQ}^2 (n) := Q(n).{eq}^2(n)$, it has the prefix\vspace{0.15cm}\\
\hspace*{0.4cm}
$Q(n) := \exists \{ x_1, y_1, . . . , x_n, y_n \} \forall \{ u_1, v_1, . . . , u_n, v_n \} \exists \{ t_{i,j} : i , j \in [n] \},$\vspace{0.15cm}\\
and the matrix ${eq}^2 (n)$ consisting of the clauses: \hspace{2.2cm} Labels:\vspace{0.15cm}\\
\hspace*{0.3cm}
\begin{minipage}{0.98\textwidth}
$\{ x_i, y_j, u_i, v_j, t_{i,j} \} , \{ x_i, \overline{y_j}, u_i, \overline{v_j}, t_{i,j} \} , ~$ for $i,j \in [n]$, \hspace{1.2cm} $C_{i,j} , C_{i,j}' $ \vspace{0.15cm}\\ 
$\{ \overline{x_i}, y_j, \overline{u_i}, v_j, t_{i,j} \} , \{ \overline{x_i}, \overline{y_j}, \overline{u_i}, \overline{v_j}, t_{i,j} \} , ~$ for $i,j \in [n]$, \hspace{1.2cm} $D_{i,j} ,  D_{i,j}'$\vspace{0.15cm}
\end{minipage}
\hspace*{2.1cm}$(\overline{t_{i,j}} : i,j \in [n])$. \hspace{5.4cm} $T$
\end{definition}
\subsection{Short \qrat refutations of ${EQ}^2 (n)$}\label{3.1}
%The idea of this proof is to drop all universal variables under QRATU rule and obtain unit clauses of $t_{i,j}$'s by resolution steps. Then to use consecutive resolutions with clause $T$ to derive a $\bot$. 
In this section we give the first short \qrat refutation for SquaredEquality formulas.
\begin{theorem}\label{short_proof}
The SqauredEquality formulas have $\mathcal{O}(n^2)$-size \qrat refutations.
\end{theorem}
\begin{proof}
Let $n \in \mathbb{N}$,where $\mathbb{N}$ is the set of all natural numbers. We construct a refutation in 3 stages. In the first stage, we drop all the universal variables in the formulas for the reason that they are QRAT-literals. In the second stage we derive unit clauses of all $t_{i,j}$'s using 3 resolutions steps each. In the last stage we successively resolve these unit clauses with the clause $T$ and derive an empty clause.
\vspace{0.1cm}
\begin{alphaenumerate}
\item \textbf{Stage 1:} we prove the following lemma first.\vspace{0.1cm}
\begin{lemma}\label{lemma_qratu}
All universal literals are QRAT-literals in SquaredEquality formulas and can be dropped by QRATU rule.
\end{lemma}
\vspace*{-0.45cm}
\begin{proof}
Observe that in all the $4$ type of clauses (i.e. $C_{i,j} ,  C_{i,j}' , D_{i,j} ,  D_{i,j}' $) the existential literal $x_i$ is always in the same clause as the universal literal $u_i$ and the literal $\overline{x_i}$ is always in the same clause as the literal $\overline{u_i}$. Same is with the existential variable $y_j$ and universal variable $v_j$. Moreover $x_i,y_j$ are always on the left of $u_i,v_j$ in the quantifier prefix.\vspace{0.1cm}\\
Consider the $C_{i,j}$ type of clauses, they contain the universal literal $u_i$. The outer resolvents of these clauses can be with either $D_{i,j}$ or $D_{i,j}'$ which contain the literal $\overline{u_i}$. All these outer resolvents will have both $x_i$ and $\overline{x_i}$, i.e. they are a tautology. Hence all the outer resolvents are ATs, which makes $u_i$ a QRAT-literal in $C_{i,j}$. A similar argument can be made for each one of the $4$ clauses as the primary clause. Thus all the $u_i$ variables can be dropped from the formulas.\vspace{0.1cm}\\
%Similarly as observed before 
Now, consider again the $C_{i,j}$ type of clauses, they contain the universal literal $v_j$. The outer resolvents of these clauses can be with either $C_{i,j}'$ or $D_{i,j}'$ which contain the literal $\overline{v_j}$. All these outer resolvents will have both $y_j$ and $\overline{y_j}$, i.e. they are a tautology. Hence they are all ATs, that makes $v_j$ a QRAT-literal in $C_{i,j}$. Similar arguments can be made for each one of the 4 clauses. So all the $v_j$ variables can be dropped from the formulas.
\end{proof}
\vspace{-0.4cm}
Now using Lemma \ref{lemma_qratu}, we drop all universal variables in $\mathcal{O}(n^2)$ and obtain the following clauses: \hspace{7.2cm} Labels:\vspace{0.2cm}\\
\hspace*{1cm}
\begin{minipage}{0.98\textwidth}
$\{ x_i, y_j, t_{i,j} \} , \{ x_i, \overline{y_j}, t_{i,j} \} , ~$ for $i,j \in [n]$, \hspace{1.2cm} $ C_{i,j}'',  C_{i,j}''' $ \vspace{0.15cm}\\ 
$\{ \overline{x_i}, y_j, t_{i,j} \} , \{ \overline{x_i}, \overline{y_j}, t_{i,j} \} , ~$ for $i,j \in [n]$, \hspace{1.2cm} $ D_{i,j}'' ,  D_{i,j}'''$\\
\end{minipage}
\item \textbf{Stage 2:} For every ${i,j} \in [n]$, we use 3 resolution rules on the corresponding clauses  $C_{i,j}'' ,~  C_{i,j}''' ,~  D_{i,j}''~\&~ D_{i,j}''' $ to obtain the unit clause $(t_{i,j})$ as  follows:
\vspace{0.15cm}\\
\hspace*{1cm}
\begin{minipage}{0.98\textwidth}
$P_{i,j} = Res( C_{i,j}'' ,  C_{i,j}''', y_j) = \{ x_i , t_{i.j} \}$\vspace{0.15cm}\\
$Q_{i,j} = Res( D_{i,j}'' ,  D_{i,j}''', y_j) = \{ \overline{x_i} , t_{i.j} \}$\vspace{0.15cm}\\
$R_{i,j} = Res(P_{i,j},Q_{i,j}, x_i) = \{ t_{i,j} \}$ \vspace{0.15cm}
\end{minipage}
Resolution clauses are AT in \qrat (Observation~\ref{obs:qrat-simulate-res}), so the above clauses ($P_{i,j}, Q_{i,j}, R_{i,j}$ in this order) can be added directly. This stage can be done in $\mathcal{O}(n^2)$ resolution steps.

\vspace{0.15cm}
%\vspace{-0.3cm}\\
\item \textbf{Stage 3:} For every $i,j \in [n]$ we have already derived all the $n^2$ unit clauses $R_{i,j}$'s, using these clauses along with the input clause $T$ we may derive the empty clause $\bot$ in $\mathcal{O}(n^2)$ steps. 
%so using these and the clause $T$, we can easily derive the $\bot$ in $n^2$ resolutions. 
%we cumulatively resolve the corresponding unit clause with the remaining clause of $T$ i.e after every such resolution, the clause $T$ shrinks by a literal. Hence after the last of such resolutions the clause $T$ is an empty cause. 
%Even this stage can be completed in $\mathcal{O}(n^2)$ steps. 
%\vspace{-0.3cm}\\
\end{alphaenumerate}
\vspace{0.1cm}
This completes the proof. Observe that SquaredEquality formulas have $\mathcal{O}(n^2)$ clauses, hence the \qrat refutation is indeed linear in the size of the formula.
%This completes the proof and 
%as mentioned in each stage , they can be done in $\mathcal{O}(n^2)$ steps. Therefore 
%and also our \qrat refutation size is $\mathcal{O}(n^2)$ and since this is equal to the number of clauses in the SquaredEquality formulas; size of \qrat refutations for ${EQ}^2(n)$ is linear with the number of clauses.
\end{proof}
%The \qrat proof of ${EQ}^2 (n)$ is very short as all the universal variables can be dropped under QRATU rule. i.e, wherever the universal literal $u_i$ (resp $v_j$) occurs in a clause the literal $x_i$ (resp $y_j$) occurs in the same clause. Also, whenever the universal literal $\overline{u_i}$ (resp $\overline{v_j}$) occurs in a clause the literal $\overline{x_i}$ (resp $\overline{y_j}$) occurs in the same clause. Given that both $x_i$ and $y_j$ occur always to the left of $u_i$ and $v_j$, it implies that the outer-resolvents of $u_i$ and $v_j$ (resp $\overline{u_i}$ and $\overline{v_j}$) in every clause always leads to a tautology. Therefore dropping of universal variables is sound according to \qrat. Hence the first $n^2$ steps correspond to adding clauses of the form ${eq}^2$ but with only the existential literals.
%For the formulas of a particular $i,j$, it can be seen that a unit clause containing only ($t_{i,j}$) can be obtained in 3 resolution steps (2 on $y_j$ and 1 on $x_i$ or vice-versa). That counts to 3$n^2$ steps resulting in $n^2$ unit clauses of ($t_{i,j}$)'s .
%The successive resolution steps of these unit clauses with the huge clause of $(\overline{t_{i,j}} : i,j \in [n])$ will result in $\bot$ , and can be done in $n^2$ steps. Therefore the \qrat proof of ${EQ}^2$ formulas is $\mathcal{O}(n^2)$. This motivates us to work for simulation of M-Res with \qrat. 
%\section{Strategy Extraction in both Proof systems}
Now we proceed to our simulation result in the next section.
\section{\qrat polynomially simulates \mergeres}\label{4}
%Apart from refuting a QBF, \mergeres by design is a proof system which simultaneously computes strategies using merge maps. The Merge map at the conclusion is the complete strategy for universal player against all moves of the existential player. Contrarily, it has been proved that strategy extraction from a \qrat proof is hard in \cite{qrat_strategy} accounting to the failure in finding lower bounds for this proof system i.e. till date no hard problem for \qrat has been proved. 

%Therefore, we focus on simulating only the refutation of \mergeres system which are the sequence of $C_i$'s from (definition \ref{mres}) which end in $C_k = \bot$.

In \mergeres, in addition to finding a winning strategy for the universal player, the proof system also derives the empty clause $\bot$ through a sequence of sound rules.
That is at the end, \mergeres proves that the given QBF is false in two different ways simultaneously. Firstly, by providing through sound rules, a sequence of clauses with only existential literals $C_1, C_2 \dots, C_k = \bot$. Secondly, by providing a countermodel for the QBF through merge-maps.\vspace{0.15cm} \\
%{\bf We may need to add one line here regarding soundness}\\
While deriving the clauses $C_i$ in the sequence above, \mergeres consults the corresponding partial strategies presented in the hypothesis and makes sure that they meet certain criteria (Definition~\ref{mres}) to maintain soundness. Therefore, this sequence of clauses depends on the partial strategies that the \mergeres proof is building. 
However, it is sufficient for any proof system to produce either of these proof types to prove the falseness of any QBF.

So, even if a proof system $f$ can efficiently simulate through its sound rules the sequence of clauses $C_1, C_2, \dots, C_k$ then we can say that $f$ polynomially simulates \mergeres. $f$ is not required to build or consult the partial strategies built by \mergeres.

The other way of simulating \mergeres by a proof system $f$ would be to simulate the process of building the partial strategies of the \mergeres proof.

\vspace{0.15cm}
\noindent
We show that \qrat can efficiently simulate \mergeres by simulating it's sequence of clauses $C_1, C_2, \dots, C_k$, as is mentioned in the first process.

%So observe, there are two ways of simulating a given \mergeres refutation by a proof system $Q$. Either $Q$ simulates the process of building the partial strategies of the \mergeres or simulate the $\bot$- derivation steps of \mergeres. We show that \qrat can successfully simulate \mergeres through the second process. 
\begin{theorem}\label{simul}
\qrat polynomially simulates \mergeres.
\end{theorem}
\begin{proof}
Given an \mergeres refutation $\pi$ = $L_1,...,L_k$ for a false QBF $\Phi = Q.\phi$ with $X$ (resp.~$U$) as the set of existential (resp.~universal) variables, where each $L_i = (C_i,\{ {M_i}^u : u\in U \} )$, we effectively compute a \qrat refutation $\Pi$ for the QBF $\Phi = Q.\phi$ as follows: \vspace{0.1cm}
\begin{alphaenumerate}
\item \textbf{Axiom Steps:} For every axiom step in $\pi$ (say $L_i$), the following Lemma holds:\vspace{0.1cm}
\begin{lemma}\label{existential_clause}
The existential sub-clause of a clause $C \in \phi$ is AT w.r.t the QBF $\Phi$.
\end{lemma}
\vspace{-0.4cm}
\begin{proof}
Let $C= \{ e_1,.., e_n, u_1,..., u_m \}$ be an input clause in the QBF $Q.\phi$, where $e_1,..,e_n$ are existential literals and $u_1,..,u_m$ are universal literals with arbitrary order in $Q$. The existential sub-clause $C_i= \{ e_1,.., e_n \}$ is an AT w.r.t $Q.\phi$: %the unit propagation of $\overline{C_i}$ makes the clause $C = \{ u_1,...,u_m \} $.
%but all these are universal variables i.e there is always a particular assignment for them which falsifies all the literals in $C$. Thus making $C = \bot$.
since $\overline{C_i} \wedge C ~\vdash_1~ (u_1,..., u_m) =\bot$. The clause $(u_1,...,u_m)$ is the empty clause since all its literals are universal.
\end{proof}
\vspace{-0.3cm}
Using Lemma \ref{existential_clause}, we can directly add the existential sub-clause $C_i$ belonging to the axiom step $L_i$ of \mergeres proof $\pi$ in the \qrat proof $\Pi$. This can be done for all the axiom steps in order as they appear in $\pi$. %\vspace{0.1cm}
\item \textbf{Resolution Steps:} Resolution step in \mergeres is executed provided some conditions on hypothesis merge maps are met (Definition \ref{mres}).
So while simulating, \qrat only needs to simulate the resolution steps where soundness part is already taken care of by \mergeres via maintaining partial strategies through merge-maps.

\vspace{0.1cm}
%But given a valid \mergeres proof, refutation part of the resolution rule is same as a \qrc's resolution. i.e the clause $C_i$ is obtained by resolution of the previously included clauses $C_a$ and $C_b$ on a variable occurring in opposite polarities in both of them. 
The resolvent clause is known to be AT w.r.t. QBF $\Phi$, so can be directly added to $\Pi$ (Observation~\ref{obs:qrat-simulate-res}).
When \qrat simulates the last line of $\pi$ through resolution, we get the corresponding $\bot$ in $\Pi$ as well.
%The conclusion of \mergeres i.e., the empty clause $\bot$ is also obtained by resolution, so we can add $\bot$ to $\Pi$ at the end as well.\vspace{-0.3cm}\\
\end{alphaenumerate}
\vspace{0.05cm}
This completes the simulation. Observe that $\Pi$ is a valid \qrat refutation of the input formula $\Phi$.
Note that, the size of the \qrat proof is linear in the size of the corresponding \mergeres proof. %(in-fact, it is exactly = $k$) and 
Also observe that the sequence of clauses $C_1,...,C_k$ in the lines of the \mergeres are in itself a valid \qrat proof!
\end{proof}
\vspace{-0.2cm}
On the other hand, we observe that \mergeres is not powerful enough to efficiently simulate the \qrat proof system. To be precise, we have:\vspace{-0.1cm}
\begin{observation}
\mergeres cannot simulate \qrat.
\end{observation}
\vspace{-0.3cm}
\begin{proof}
There exists a family of false QBFs KBKF-lq$[n]$~\cite[Definition 3]{BalabanovWJ14} which are shown to be hard for \mergeres in \cite[Theorem 19]{hard_mres}, but easy for \qurc~\cite[Theorem 2]{BalabanovWJ14}. Since \qrat simulates \qurc\cite{ldqres_simul}, these formulas are easy for \qrat. So this concludes that \mergeres cannot simulate \qrat.
\end{proof}

\vspace{-0.3cm}
\section{Conclusions and future work}\label{5}
\qrat proof system is capable of efficiently simulating both the expansion-based QBF-solving approach, i.e., \ecalculus~\cite{ecalculus_simul} and the CDCL-based QBF-solving approach \lqrc~\cite{ldqres_simul}. It is also known that \qrat can simulate all the existing preprocessing techniques used by current QBF-solvers~\cite{qrat_paper}. In this paper, 
we show that \qrat can even strictly simulate the new proof system which builds partial strategies into proofs, that is, the \mergeres proof system~\cite{mres_paper}. Thus extending the importance of \qrat among QBF proof systems.
Also, we have given a short \qrat refutation for the SquaredEquality formulas introduced in \cite{mres_paper}.\vspace{0.12cm}\\
%we show that \qrat can efficiently simulate even the \mergeres proof system~\cite{mres_paper}, a new QBF proof system which builds partial strategies into proofs. Thus extending the importance of \qrat among QBF proof systems.
%We established the simulation result by showing that \qrat can simulate efficiently the sequence of existential clauses derived by \mergeres. It is known that, for false QBFs, \qrat does not admit strategy extraction unless P $=$ PSPACE~\cite{ChewC20}. Therefore, \qrat simulation of building the partial strategies in the \mergeres proof system is pointless. i.e. building strategies in \qrat is not feasible. 
\noindent
Work in this domain still has many interesting open problems; we would like to mention a few of the same: 

Although \mergeres was inspired from \lqrc, it is still open if they are incomparable or if one can simulate the other. 

\qrat simulation of \ecalculus has been proven in~\cite{ecalculus_simul}, but it is still open whether or not \qrat can simulate it's powerful variant the \irc proof system? Note that \irc cannot simulate \qrat proof system since the former is incomparable with \lqrc and \qrat simulates \lqrc~\cite{ldqres_simul}. For the complexity landscape of these systems, refer Figure~\ref{fig1} in this paper.

Given a false QBF $\Phi = Q.\phi$, \mergeres builds a winning strategy of the universal player by design. In case, $\Phi$ has a computationally hard winning strategy, \mergeres refutation of $\Phi$ is also going to be large. As a result, proving lower bound for such QBFs in \mergeres is easy. On the other hand, \qrat does not admit strategy extraction for false QBFs unless P $=$ PSPACE~\cite{ChewC20}. That is, there exists a family of false QBFs (the Select Formulas from~\cite[Section 4.1]{ChewC20} which are easy for \qrat but have computationally hard universal winning strategies provided P $\neq$ PSPACE. As a result, establishing lower bound results for false QBFs in \qrat using the strategy extraction technique is not possible. In fact, proving a lower bound result in \qrat is still open. It should be noted that \qrat admits strategy extraction for true QBFs~\cite{HeuleSB14},
therefore strategy extraction can still be used to prove \qrat lower bounds for true QBFs.

\bibliography{lipics-v2021-sample-article}

\end{document}